\documentclass[aps,prl,twocolumn,superscriptaddress,nofootinbib]{revtex4-2}
 
\usepackage[utf8]{inputenc}
\usepackage[T1]{fontenc}
\usepackage{lmodern}
\usepackage{amsmath,amssymb,amsthm,mathtools,bm}
\usepackage{graphicx}
\usepackage[dvipsnames]{xcolor}
\usepackage[colorlinks=true, linkcolor=MidnightBlue, citecolor=MidnightBlue, urlcolor=MidnightBlue]{hyperref}
\usepackage{bbm}
\usepackage{mathrsfs}
\usepackage{braket}
\usepackage{ytableau}
\ytableausetup{boxsize=0.55em}

\usepackage{tikz}
\usetikzlibrary{decorations.pathreplacing,calc}
 
\newtheorem{theorem}{Theorem}
\newtheorem{lemma}[theorem]{Lemma}
\newtheorem{proposition}[theorem]{Proposition}

\newtheorem{corollary}[theorem]{Corollary}

\DeclareMathOperator{\conv}{conv}

\DeclareMathOperator{\supp}{supp}
 
\newcommand{\cH}{\mathcal H}

\newcommand{\cY}{\mathcal Y}
\newcommand{\cQ}{\mathcal Q}
\newcommand{\cP}{\mathcal P}

\newcommand{\sqtwo}{s_2}
\newcommand{\proj}[1]{\ket{#1}\!\bra{#1}}
\newcommand{\Prob}{\mathsf{Prob}}
\newcommand{\Up}{\mathsf{Up}}
\newcommand{\Down}{\mathsf{Down}}
\newcommand{\up}{\mathcal{U}}
\newcommand{\down}{\mathcal{D}}
 
\begin{document}

\title{Irreducible Architectures of Multipartite Entanglement}

\author{Augusto Smerzi}
\email{augustosmerzi@sztu.edu.cn}
\affiliation{College of Engineering Physics, Shenzhen Technology University, Shenzhen 518118, China}
\author{Manuel Gessner}
\email{manuel.gessner@uv.es}
\affiliation{Instituto de F\'isica Corpuscular (IFIC), CSIC-Universitat de Val\`encia and Departament de F\'isica Te\`orica, UV, Avinguda Vicent Andr\'es Estell\'es 19, E-46100 Burjassot (Valencia), Spain}

\date{\today}

\begin{abstract}
Multipartite entanglement is commonly characterized by scalar notions
such as separability and entanglement depth, which do not resolve the
distribution of entangled cluster sizes. For mixed states, we introduce
formation profiles that assign weights to the entanglement
architectures appearing in pure-state decompositions. We show that after discarding
every profile that admits a strictly weaker feasible replacement, the
remaining irreducible structure need not be unique: a four-qubit
example exhibits a continuous family of incomparable irreducible
profiles. Monotone functions of the architectures recover widely used scalar
quantifiers as special cases, whereas the full profile geometry retains
additional information, including the minimum weight that every
decomposition must assign outside a chosen architectural class. Finally,
we derive experimentally accessible bounds on these weights from convex
witnesses, including the quantum Fisher information, thereby connecting
detailed formation structure with practical entanglement certification.
\end{abstract}

\maketitle

\textbf{Introduction.} Multipartite entanglement is a characteristic feature of interacting quantum
systems and a resource for quantum sensing, computation, communication,
and networks
\cite{GuhnePhysRep2009,PezzeRMP2018,FriisNatRevPhys2019}.
Today's experiments are able to generate and certify multipartite
entanglement across a broad range of platforms,
including ultracold atoms, trapped ions, superconducting qubits, and
networked quantum nodes
\cite{ZhangPRL2023,BohnetScience2016,CaoNature2023,
PompiliScience2021}.
A remaining central challenge is to determine which entangled
cluster-size patterns are required to form a mixed state, beyond a single
separability or depth label.

Partitions provide an intuitive description of this architecture.
Young diagrams represent a partition by rows whose lengths give the
cluster sizes~\cite{SzalayJPA2018,SzalayQuantum2019}. For four parties, for example,
\((4)\simeq \ydiagram{4}\) represents a single four-party cluster with genuine multipartite entanglement
(GME),
\((3,1)\simeq \ydiagram{3,1}\) a three-party cluster
with a spectator, and
\((2,2)\simeq \ydiagram{2,2}\) two entangled pairs. The refinement order between such diagrams organizes
partial separability and related notions, including
\(k\)-separability, \(k\)-producibility, entanglement depth,
stretchability, and more general diagram-based classifications~\cite{DuerPRL1999,SzalayJPA2018,SzalayQuantum2019,SzalayQuantum2025}.
Optimizable witnesses can probe multipartite partition structure
without full state tomography~\cite{LuPRX2018}, while the QFI accesses
selected Young-diagram properties through collective observables
\cite{PezzeSmerziPRL2009,HyllusPRA2012,TothPRA2012,RenPRL2021}.

For pure states, the cluster architecture is unambiguous: every state
has a unique finest tensor-product factorization. The situation is
fundamentally different for mixed states. The same density matrix can
be formed from pure states with inequivalent architectures, and there
need not be a preferred decomposition. Existing classifications
largely address this ambiguity through membership questions, whether
a state admits a decomposition within a prescribed family
\cite{SzalayJPA2018}, or by assigning a single scalar value to every
diagram and minimizing a corresponding depth or average cost
\cite{SzalayQuantum2025}. Such projections are powerful,
but they discard how different architectures can be combined within
one decomposition and whether independently optimal descriptions are
mutually compatible.

Here, we characterize a mixed state by the complete set of exact
cluster-size weight distributions attainable in its pure-state
decompositions. Removing profiles that admit a strictly weaker feasible
replacement yields a Pareto frontier of irreducible architectures, which
need not consist of a single profile. A four-qubit state already exhibits
 a continuous trade-off between the incomparable architectures \((2,2)\)
and \((3,1)\), while every formation requires an unavoidable exact-\((4)\)
contribution. We then show that convex witnesses constrain these
distributions and
convert membership tests into quantitative lower bounds on the fraction
of every decomposition that must lie outside a chosen architectural
class. For the QFI, we evaluate these experimentally accessible bounds
for one-axis-twisted spin states. Our approach therefore resolves the
joint trade-off between inequivalent formation architectures, going
beyond scalar depth and membership criteria.

\textbf{Exact sectors and formation profiles.} We consider a finite-dimensional \(N\)-partite Hilbert space
\(\cH_N=\bigotimes_{j=1}^N\cH^{(j)}\) and denote by
\(\mathcal S(\cH_N)\) its density-operator state space. A Young diagram
\(\Lambda=(\lambda_1,\ldots,\lambda_{h(\Lambda)})\vdash N\) is an
integer partition of \(N\), with
\(\lambda_1\ge\cdots\ge\lambda_{h(\Lambda)}\ge1\) and
\(\sum_i\lambda_i=N\). Its height \(h(\Lambda)\) is the number of
clusters, while its width \(\lambda_1\) is the size of the largest
cluster and underlies the usual notion of entanglement depth. We denote
the finite set of all such diagrams by \(\cY_N\). A Young diagram
records only the sizes of the exact tensor factors, not which labeled
parties belong to them.

Every pure state \(\ket{\psi}\in\cH_N\) has a unique finest partition of
the parties into tensor-product factors; the ordered block sizes define
its exact Young diagram \(\Lambda_\psi\), as shown in the Supplemental
Material~\cite{SM}. For fixed \(\Lambda\), we define the exact pure sector and its
convex hull as
\begin{align}
\mathscr P_\Lambda
=
\{\proj{\psi}:\Lambda_\psi=\Lambda\},
\qquad
\mathfrak E_\Lambda
=
\conv(\mathscr P_\Lambda).
\label{eq:exact-sector}
\end{align}
A state in \(\mathfrak E_\Lambda\) can be formed entirely from pure
states of exact type \(\Lambda\), but it need not itself retain that
factorization structure after mixing. For example,
\begin{align}
\frac{1}{2}\left(\proj{\Phi^+}+\proj{\Phi^-}\right)
=
\frac{1}{2}\left(\proj{00}+\proj{11}\right),
\end{align}
so a mixture of two exact-\((2)\) Bell states can be fully separable.

Starting from any pure-state decomposition of \(\rho\), one can group
together all components with the same exact Young diagram. This yields
\begin{align}
\rho
=
\sum_{\Lambda\in\cY_N}
q_\Lambda\rho_\Lambda,
\qquad
\rho_\Lambda\in\mathfrak E_\Lambda.
\label{eq:sector-grouping}
\end{align}
Here \(q_\Lambda\ge0\) is the total weight assigned to architecture
\(\Lambda\), with \(\sum_\Lambda q_\Lambda=1\). The vector
\(q=(q_\Lambda)_{\Lambda\in\cY_N}\in\Prob(\cY_N)\) is therefore a
probability distribution over the Young diagrams, which we call an
exact-architecture profile of \(\rho\).

A given decomposition produces only one such profile, whereas the same
mixed state generally admits many inequivalent decompositions and hence
many different profiles. We therefore define
\begin{align}
\cQ(\rho)
=
\{
q\in\Prob(\cY_N):
\rho=
\sum_{\Lambda\in\cY_N}
q_\Lambda\rho_\Lambda,\
\rho_\Lambda\in\mathfrak E_\Lambda
\},
\label{eq:Qex}
\end{align}
which collects all exact-architecture profiles compatible with
\(\rho\). The set \(\cQ(\rho)\) is convex, since mixing two
decompositions of the same state mixes their profiles in the same
proportions.

\textbf{Profile order and irreducible architecture.} Young diagrams are partially ordered by refinement
\cite{SzalayJPA2018,SzalayQuantum2019,SzalayQuantum2025}. We write
\(\Lambda\preceq\Gamma\) if \(\Lambda\) can be obtained from \(\Gamma\)
by splitting one or more clusters. Thus \(\Lambda\) is architecturally
no stronger than \(\Gamma\): it contains the same parties arranged into
smaller irreducible clusters. For example,
\((2,2)\preceq(4)\) and \((3,1)\preceq(4)\), whereas \((2,2)\) and
\((3,1)\) are incomparable.

We denote by \(\Up(\cY_N)\) and \(\Down(\cY_N)\) the families of
subsets of \(\cY_N\) that are upward- and downward-closed under the
refinement order, respectively: if
\(\Lambda\in \up\in\Up(\cY_N)\), then every
\(\Gamma\succeq\Lambda\) also belongs to \(\up\), with the order reversed
for \(\down\in\Down(\cY_N)\). The profile order is then
\begin{align}
q\preceq r
\quad\Longleftrightarrow\quad
\sum_{\Lambda\in \up}q_\Lambda
\leq
\sum_{\Lambda\in \up}r_\Lambda
\quad
\forall\,\up\in\Up(\cY_N).
\label{eq:profile-order}
\end{align}
Hence \(q\preceq r\) means that \(q\) never places more weight
than \(r\) in any upward region of architecture space. In this
sense, \(q\) requires no stronger multipartite cluster structure than
\(r\).

The irreducible architecture of \(\rho\) is obtained by discarding every
feasible profile that can be replaced by a strictly weaker one. The
remaining set is
\begin{align}
\cP(\rho)
=
\left\{
q\in\cQ(\rho):
\nexists\,r\in\cQ(\rho)
\text{ with }r\prec q
\right\}.
\label{eq:Pareto}
\end{align}
 As proved in the Supplemental Material, every feasible profile
\(q\in\cQ(\rho)\) admits a Pareto-minimal replacement
\(r\in\cP(\rho)\) with \(r\preceq q\); hence the
frontier is nonempty. Since the diagram order is only partial, there
need not be a unique weakest profile. The resulting set of non-improvable solutions is known
in multiobjective optimization as a Pareto frontier
\cite{MiettinenBOOK1999}. It contains all formation profiles that are minimal with respect to the
full architecture order, without imposing an arbitrary ranking
between incomparable Young diagrams.

Standard multipartite classifications arise by assigning a monotone scalar
function \(f:\cY_N\to\mathbb R_{\geq0}\) to each diagram, with
\(f(\Lambda)\leq f(\Gamma)\) whenever
\(\Lambda\preceq\Gamma\). Important examples include
\(f(\Lambda)=N-h(\Lambda)\), associated with
\(k\)-separability
\cite{DuerPRL1999,GuhneNJP2010},
\(f(\Lambda)=\lambda_1\), leading to
\(k\)-producibility or entanglement depth
\cite{SorensenPRL2001,PezzeRMP2018}, as well as more general functions that
probe the full diagram, such as the squareability
\(f(\Lambda)=\sqtwo(\Lambda)=\sum_i\lambda_i^2\)
\cite{SzalayQuantum2019,SzalayQuantum2025}.
Each such choice projects the full profile geometry onto one scalar
scale and may therefore order diagrams that are incomparable at the
architectural level.

For mixed states, each \(f\) gives rise to two natural scalar extensions~\cite{SzalayQuantum2025}:
the minimum average value in a formation and the minimum largest value
that must appear
\begin{align}
F_f(\rho)
&=
\inf_{q\in\cP(\rho)}
\sum_{\Lambda\in\cY_N}q_\Lambda f(\Lambda),
\label{eq:formation}
\\
D_f(\rho)
&=
\inf_{q\in\cP(\rho)}
\max_{\Lambda:q_\Lambda>0}f(\Lambda).
\label{eq:depth}
\end{align}
Here, the usual pure-state decomposition definitions are reduced to
optimizations over the Pareto frontier; the equivalence is shown in the
Supplemental Material~\cite{SM}. The frontier therefore contains all scalar
formation- and depth-type quantifiers as projections, while retaining
the trade-offs between inequivalent architectures that any single
\(f\) necessarily discards. Determining it remains a challenging mixed-state
decomposition problem, but the examples below show that nontrivial
frontier structure already arises in very small systems.

\textbf{A four-qubit state with a continuous frontier.} Consider the orthogonal states
\begin{align}
\ket{A}
&=
\ket{\Phi^+}_{12}\otimes\ket{\Phi^+}_{34},
&
\ket{B}
&=
\ket{\mathrm{GHZ}^+}_{123}\otimes\ket{-}_4,
\end{align}
of exact types \((2,2)\) and \((3,1)\), respectively, and
\begin{align}
\rho_{\rm P}
=
\frac{1}{2}\proj{A}
+
\frac{1}{2}\proj{B}
+
\frac{1}{4}\left(\ket{A}\!\bra{B}+\ket{B}\!\bra{A}\right).
\label{eq:rhoP}
\end{align}
Here \(\ket{\Phi^+}=(\ket{00}+\ket{11})/\sqrt2\),
\(\ket{\mathrm{GHZ}^+}=(\ket{000}+\ket{111})/\sqrt2\), and
\(\ket{-}=(\ket0-\ket1)/\sqrt2\).

As shown in the Supplemental Material~\cite{SM}, the support of \(\rho_{\rm P}\)
contains only the exact sectors \((2,2)\), \((3,1)\), and \((4)\).
Writing \(q=(x,y,1-x-y)\), with \(x=q_{(2,2)}\),
\(y=q_{(3,1)}\), and \(z=1-x-y=q_{(4)}\), its Pareto frontier is given by all $q$ with $0\le x,y\le\frac{3}{8}$ satisfying
\begin{align}
\left(\frac{1}{2}-x\right)
\left(\frac{1}{2}-y\right)
=
\frac{1}{16}.
\label{eq:frontier-curve}
\end{align}
Thus the irreducible architecture is a continuous family of
incomparable formations trading exact \((2,2)\) against exact
\((3,1)\) weight, see Fig.~\ref{fig:four-qubit-frontier}. Every frontier
profile satisfies \(q_{(4)}\ge1/2\), with equality at \(x=y=1/4\).
Since every feasible profile admits a no-stronger Pareto-minimal
representative, and this replacement cannot increase the exact-\((4)\)
weight, the same lower bound holds for every formation of \(\rho_{\rm P}\):
at least one-half of its total weight must be assigned to pure genuinely
four-partite entangled states.

\begin{figure}[t]
\centering
\includegraphics[width=\columnwidth]{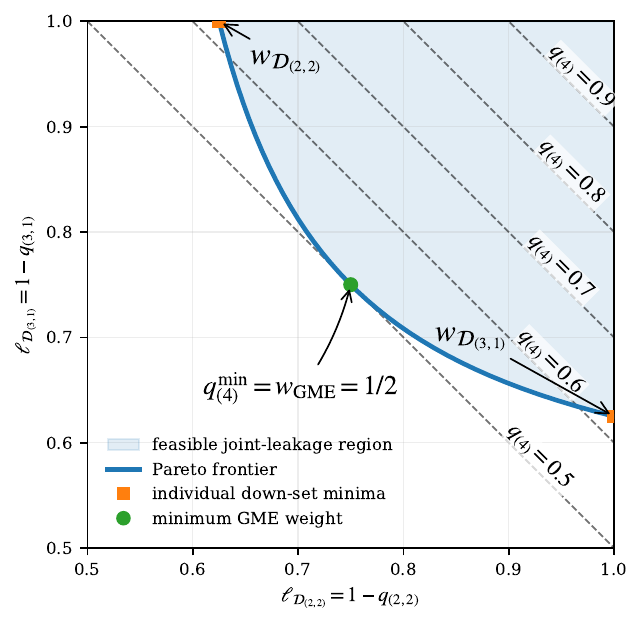}
\caption{
Joint-leakage geometry of \(\rho_{\rm P}\). The horizontal and vertical
axes are
\(\ell_{\down_{(2,2)}}=1-q_{(2,2)}\) and
\(\ell_{\down_{(3,1)}}=1-q_{(3,1)}\), respectively.
The shaded region contains all feasible leakage pairs, while the solid
boundary is their Pareto frontier. Its endpoints separately minimize
the two incomparable leakages, and no single decomposition attains both
minima. The symmetric point satisfies
\(q_{(2,2)}=q_{(3,1)}=1/4\) and minimizes the exact-\((4)\) weight,
yielding \(q_{(4)}^{\min}=1/2\).
}
\label{fig:four-qubit-frontier}
\end{figure}

\textbf{Architectural weights.} The profile geometry contains more information than the scalar formation
and depth quantifiers introduced above. As the preceding example already
illustrates, it can quantify how much formation weight must lie outside
a prescribed architectural class. For any nonempty down-set
\(\down\in\Down(\cY_N)\), we define
\begin{align}
w_\down(\rho)
=
\inf_{q\in\cQ(\rho)}
\sum_{\Lambda\notin \down}q_\Lambda.
\label{eq:downset-weight}
\end{align}
This equals the resource weight relative to
\(\mathcal F_\down=\conv(\bigcup_{\Lambda\in \down}\mathscr P_\Lambda)\)
\cite{DucuaraPRL2020}, and therefore generalizes the best separable
approximation~\cite{LewensteinPRL1998}. Equivalently, it
is the smallest \(p\) such that
\(\rho=(1-p)\sigma_\down+p\tau\), with
\(\sigma_\down\in\mathcal F_\down\) and \(\tau\) an arbitrary state;
moreover, \(w_\down(\rho)=0\) exactly when
\(\rho\in\mathcal F_\down\)~\cite{SM}.

Scalar depth thresholds are contained as the special cases
\begin{align}\label{eq:fsdownset}
\down_{f,s}=\{\Lambda\in\cY_N:f(\Lambda)\leq s\},
\qquad s\in f(\cY_N),
\end{align}
which are down-sets associated with an architectural monotone \(f\). The
corresponding tail weight \(w_{\down_{f,s}}\) quantifies how much
formation weight above \(s\) is unavoidable, whereas \(D_f\) records
only whether all such weight can be eliminated. More general down-sets need not coincide with threshold sets of any of
the standard scalar monotones, such as \(N-h(\Lambda)\), width, or squareability. For the biseparable down-set
\(\down_{\rm bisep}=\cY_N\setminus\{(N)\}\), it identifies the GME
weight \(w_{\rm GME}:=w_{\down_{\rm bisep}}\).

The profile set also retains compatibility information that is lost
when several weights are optimized independently.
For nonempty down-sets \(\bm{\down}=(\down_1,\ldots,\down_m)\), define the jointly attainable leakage region
\begin{align}
\mathcal L_{\bm \down}(\rho)
=
\left\{
\bigl(\ell_{\down_1}(q),\ldots,\ell_{\down_m}(q)\bigr):
q\in\cQ(\rho)
\right\},
\label{eq:joint-leakage-region}
\end{align}
with $\ell_\down(q)=\sum_{\Lambda\notin \down}q_\Lambda$. Each \(w_{\down_a}\) is the minimum of one coordinate, but these minima may
belong to different decompositions and need not be jointly attainable.
The region \(\mathcal L_{\bm \down}(\rho)\), or its Pareto boundary, records
their complete trade-off.

The leakage region provides a projection of the typically high-dimensional profile set. The shaded region in Fig.~\ref{fig:four-qubit-frontier} shows
\(\mathcal L_{(\down_{(2,2)},\down_{(3,1)})}(\rho_{\rm P})\), where we write \(\down_\Lambda=\{\Gamma\in\cY_N:\Gamma\preceq\Lambda\}\) for the principal down-set generated by \(\Lambda\).
Its two endpoints separately maximize the weights of the incomparable
architectures \((2,2)\) and \((3,1)\), and hence minimize the
corresponding down-set leakages. No single decomposition realizes both
optima. The symmetric point instead minimizes the weight outside the
full biseparable down-set and gives
\(w_{\rm GME}(\rho_{\rm P})=q_{(4)}^{\min}=1/2\). The same frontier
therefore contains both the individual architectural weights and their
joint compatibility.

\textbf{Witness constraints.} Let \(W\) be a real-valued convex
witness functional with finite exact-sector 
upper bounds
\(W_{\max}(\Lambda)=
\sup_{\sigma\in\mathfrak E_\Lambda}W[\sigma]\). Every feasible profile
then satisfies
\begin{align}
W[\rho]
\leq
\sum_{\Lambda\in\cY_N}
q_\Lambda W_{\max}(\Lambda).
\label{eq:witness-profile}
\end{align}
Thus an experimentally observed value excludes all
profiles whose weighted sector upper bounds are incompatible with the
measurement. Equation~\eqref{eq:witness-profile} yields several useful criteria for
different aspects of the profile geometry.
First, if
\(W_{\max}(\Lambda)\leq a+b f(\Lambda)\), with \(b>0\), then the formation-type quantifier~(\ref{eq:formation}) is bounded by
\begin{align}
F_f(\rho)
\geq
\max\left\{
0,\frac{W[\rho]-a}{b}
\right\}.
\label{eq:formation-witness-bound}
\end{align}

Second, for any nonempty \(\down\in\Down(\cY_N)\), defining
\(W_\down=\max_{\Lambda\in\down}W_{\max}(\Lambda)\) leads to the membership criterion \(W[\rho]>W_\down
\Rightarrow
\rho\notin\mathcal F_\down\).
For the scalar-generated down-sets~(\ref{eq:fsdownset}), this becomes the usual
depth criterion
\(W[\rho]>W_{\down_{f,s}}\Rightarrow D_f(\rho)>s\).

Finally, the same threshold can be made quantitative~\cite{BrandaoPRA2005,SunPRL2024,MatheQuantum2026}. When \(W_\down<W_*<\infty\), with
\(W_*=\max_{\Lambda\in\cY_N}W_{\max}(\Lambda)\), separating
the profile weight inside and outside \(\down\) yields
\begin{align}
w_\down(\rho)
\geq
\max\left\{
0,\frac{W[\rho]-W_\down}{W_*-W_\down}
\right\}.
\label{eq:general-weight-bound}
\end{align}
Thus \(W_\down\) determines whether the class is excluded, whereas the
additional global upper bound \(W_*\) converts the excess above this
threshold into a lower bound on the unavoidable formation weight
outside \(\down\). Formation, depth, and architectural weight therefore
arise as distinct consequences of the same constraint~(\ref{eq:witness-profile}) on the feasible
profile set.

\textbf{Quantum Fisher information.} We now apply this construction to
the QFI, a standard metrological witness of multipartite entanglement
\cite{HyllusPRA2012,TothPRA2012,PezzeRMP2018,RenPRL2021}.
For a collective generator
\(\hat H=\sum_{j=1}^N h^{(j)}\), where every local term has spectral
width \(\Delta h>0\), the exact-sector upper bound is determined by the
squareability,
\(F_{Q,\max}(\Lambda)
=(\Delta h)^2\sqtwo(\Lambda)\), where
\(\sqtwo(\Lambda)=\sum_i\lambda_i^2\). The bound is attained by products
of GHZ-type states built from local extremal eigenvectors on the
corresponding clusters~\cite{RenPRL2021}. The general witness results then
recover the known formation bound
\(F_{\sqtwo}(\rho)\geq F_Q[\rho,\hat H]/(\Delta h)^2\) and the
associated squareability-depth criteria \cite{SzalayQuantum2025}.

For any nonempty proper down-set
 \(\down\in\Down(\cY_N)\), let
\(s_\down=\max_{\Lambda\in\down}\sqtwo(\Lambda)\). Since then
\(s_\down<N^2\), Eq.~\eqref{eq:general-weight-bound} gives
\begin{align}
w_\down(\rho)
\geq
\max\left\{
0,
\frac{F_Q[\rho,\hat H]/(\Delta h)^2-s_\down}{N^2-s_\down}
\right\}.
\label{eq:QFI-downset-weight}
\end{align}
Crossing the threshold
\(F_Q[\rho,\hat H]/(\Delta h)^2>s_\down\) excludes
\(\mathcal F_\down\); the excess above it quantifies the minimum weight
that every formation must assign outside this class. Because the QFI upper bound is fixed entirely by squareability, it
distinguishes down-sets only through \(s_\down\). Thus \(\down\) and
the generally larger scalar-generated down-set
\(\down_{\sqtwo,s_\down}\) yield the same QFI exclusion criterion, although their actual architectural
weights may differ.

\begin{figure}[t]
\centering
\includegraphics[width=\columnwidth]{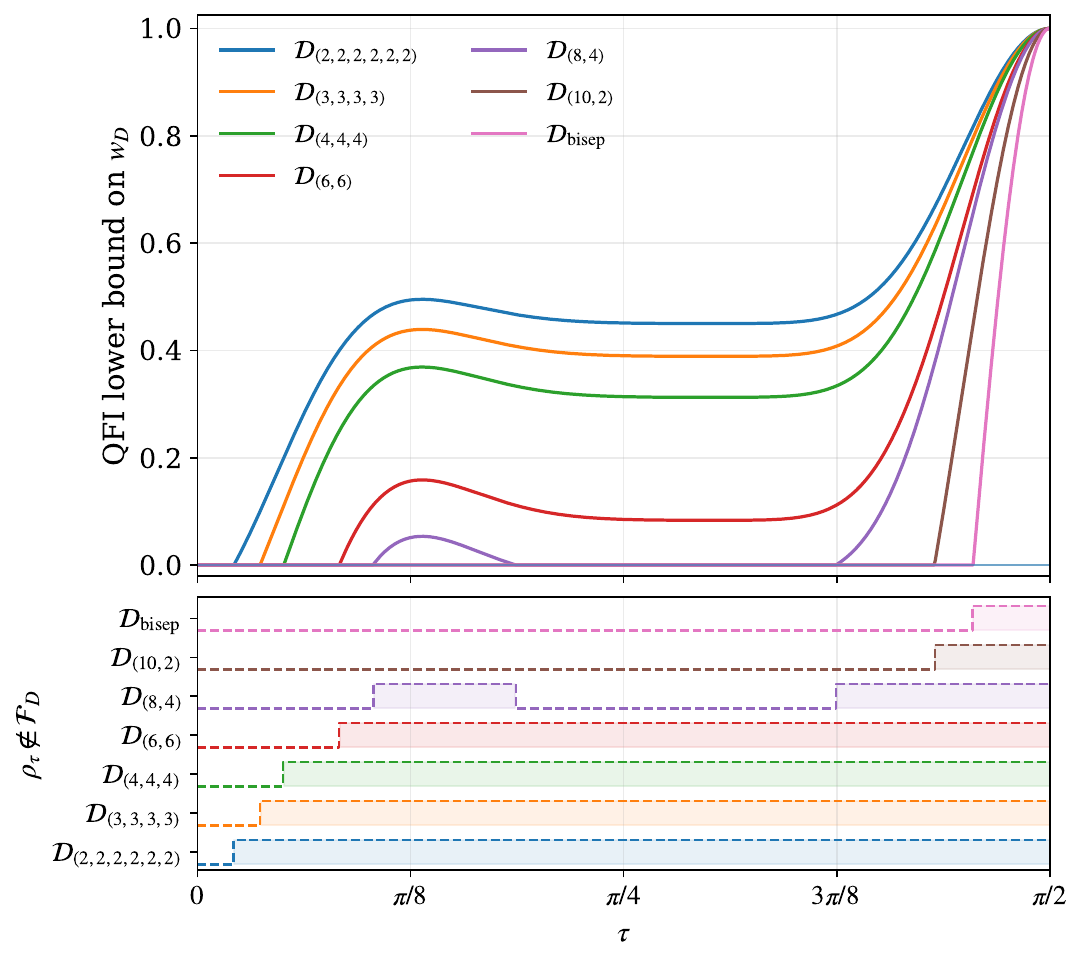}
\caption{
QFI lower bounds on architectural weights during ideal one-axis twisting
for \(N=12\). The upper panel shows QFI lower bounds on the unavoidable weight outside
representative down-sets, while the lower panel retains only the
corresponding exclusion thresholds. At \(\tau=\pi/2\), the state has
exact architecture \((12)\), and every proper-down-set weight equals
one.
}
\label{fig:oat-weights}
\end{figure}

Figure~\ref{fig:oat-weights} illustrates the quantitative weight bounds for the
one-axis-twisted states
\cite{KitagawaPRA1993,PezzeRMP2018},
\(\ket{\psi_\tau}=e^{-i\tau J_z^2}\ket{+x}^{\otimes N}\), where
\(J_z=\frac{1}{2}\sum_j\sigma_z^{(j)}\) and
\(\ket{+x}^{\otimes N}\) is the maximal-\(J_x\) eigenstate. The QFI is optimized over collective rotation
directions; details are given in the Supplemental Material~\cite{SM}. The lower panel shows when each architectural class is excluded, while
the upper panel uses the same QFI to quantify the unavoidable weight
outside it. This includes principal down-sets generated by incomparable
diagrams, such as \((4,4,4)\) and \((6,6)\), rather than only successive
levels of one scalar hierarchy. The curve for
\(\down_{\rm bisep}\) bounds
\(w_{\rm GME}=w_{\down_{\rm bisep}}\). The pure GHZ-type state at \(\tau=\pi/2\) has exact architecture \((N)\), so its weight outside every proper
down-set is one, explaining the common endpoint of all curves.

The exact QFI need not be determined experimentally to use this method: any certified lower
bound can be inserted into Eq.~\eqref{eq:QFI-downset-weight}. In the Supplemental Material, we apply the same construction to
nonlinear spin-squeezing bounds, which have been measured in
one-axis-twisted superconducting-qubit states
\cite{XuPRL2022}, and study their behavior under realistic dephasing noise~\cite{BaamaraPRL2021,BaamaraCRPhys2022}.

\textbf{Conclusion.} We introduced a formation-level description of multipartite
entanglement architecture based on feasible probability profiles over
exact Young-diagram sectors. Its Pareto frontier removes all
architecturally dominated formations while retaining the trade-offs
between incomparable cluster structures. The four-qubit example shows
that this irreducible architecture can be a continuous family rather
than a single diagram, scalar depth, or optimal decomposition.

Down-sets define architectural free classes, and their resource weights
quantify the fraction of every formation that must lie outside them.
Several such weights may be individually attainable but mutually
incompatible, a distinction retained by the joint profile geometry.
Convex witnesses constrain this geometry and convert standard
membership thresholds into quantitative formation bounds. The
QFI provides experimentally accessible bounds through squareability,
with conventional formation and depth criteria recovered as scalar
projections.

Our results isolate the cluster-size architecture layer of multipartite
entanglement. Standard notions such as \(k\)-separability,
\(k\)-producibility, stretchability, and squareability arise by
projecting the profile geometry onto a scalar function. The present
theory does not identify which labeled parties occupy a cluster, nor the
amount or type of entanglement within a cluster. Extending the profile
to labeled partitions and state-internal entanglement data would provide
a finer formation-level description.

\begin{acknowledgments}
A.S. acknowledges support by the National Natural Science Foundation of China (Grant No. W2531008) and the Peacock Plan. This work was funded by the project PID2023-152724NA-I00, with funding from MCIU/AEI/10.13039/501100011033 and FSE+, by the project CNS2024-154818 with funding by MICIU/AEI /10.13039/501100011033, by the project CIPROM/2022/66 with funding by the Generalitat Valenciana, 
and by the CSIC Interdisciplinary Thematic Platform (PTI+) on Quantum Technologies (PTI-QTEP+). This work is supported through the project CEX2023-001292-S funded by MCIU/AEI.
\end{acknowledgments}

\clearpage
\onecolumngrid

\section*{Supplemental Material}

\setcounter{section}{0}
 \setcounter{equation}{0}
\setcounter{figure}{0}
\setcounter{theorem}{0}
 \renewcommand{\thesection}{S\arabic{section}}
 \renewcommand{\theHsection}{S\arabic{section}}
\renewcommand{\theequation}{S\arabic{equation}}
\renewcommand{\theHequation}{S\arabic{equation}}
\renewcommand{\thefigure}{S\arabic{figure}}
\renewcommand{\theHfigure}{S\arabic{figure}}
 \renewcommand{\thetheorem}{S\arabic{theorem}}
\renewcommand{\theHtheorem}{S\arabic{theorem}}

\section{Pure-state factorization and architecture sectors}

A labeled partition \(\pi\) of the parties is a factorization partition of \(\ket\psi\) if \(\ket\psi\) factorizes as a tensor product over the blocks of \(\pi\). The associated Young diagram is the ordered list of block sizes.

\begin{lemma}[Unique finest factorization]\label{lma:1}
Every pure state \(\ket\psi\in\cH_N\) has a unique finest factorization partition. Hence it has a unique exact Young diagram \(\Lambda_\psi\).
\end{lemma}

\begin{proof}
If \(\ket\psi\) factorizes with respect to two labeled partitions \(\pi\) and \(\pi'\), then its reduced density matrices on the blocks of both partitions are pure: Indeed, let \(C=A\cap B\) be a block of \(\pi\wedge\pi'\), where \(A\in\pi\)
and \(B\in\pi'\). Write \(A=CX\) and \(B=CY\), and denote all
remaining parties by \(R\), as illustrated in Fig.~\ref{fig:common-refinement}.
The two factorizations can then be written as
\(\ket\psi=\ket\alpha_{CX}\otimes\ket\beta_{YR}
=\ket\gamma_{CY}\otimes\ket\delta_{XR}\). The first decomposition implies
that \(\rho_{CX}\) is pure. The second decomposition implies, after
tracing out \(Y\) and \(R\), that \(\rho_{CX}=\rho_C\otimes\rho_X\).
Thus \(\rho_{CX}\) is a pure product state, and hence \(\rho_C\) is
pure. Since this holds for every block \(C\) of the common refinement,
all blocks of \(\pi\wedge\pi'\) have pure reduced states. Because the
global state is pure, this is equivalent to factorization with respect
to \(\pi\wedge\pi'\). Iterating over all factorization partitions gives a unique finest one. Ordering the block sizes of this finest partition gives \(\Lambda_\psi\).
\end{proof}

\begin{figure}[h]
\centering
\begin{tikzpicture}[
    box/.style={draw, minimum height=0.55cm, minimum width=1.05cm, inner sep=0pt},
    active/.style={box, fill=gray!18},
    guide/.style={gray!60, thick},
    brace/.style={decorate, decoration={brace, amplitude=4pt}},
    font=\small
]

\def\yTop{1.2}
\def\yBot{0}

\node at (-0.7,\yTop) {\(\pi\)};
\node at (-0.7,\yBot) {\(\pi'\)};

\node[box] (tRone) at (0,\yTop) {\(R_1\)};
\node[active, anchor=west] (tX) at (tRone.east) {\(X\)};
\node[active, anchor=west] (tC) at (tX.east) {\(C\)};
\node[box, anchor=west] (tY) at (tC.east) {\(Y\)};
\node[box, anchor=west] (tRtwo) at (tY.east) {\(R_2\)};

\node[box] (bRone) at (0,\yBot) {\(R_1\)};
\node[box, anchor=west] (bX) at (bRone.east) {\(X\)};
\node[active, anchor=west] (bC) at (bX.east) {\(C\)};
\node[active, anchor=west] (bY) at (bC.east) {\(Y\)};
\node[box, anchor=west] (bRtwo) at (bY.east) {\(R_2\)};

\draw[brace]
  ($(tX.north west)+(0,0.16)$) -- ($(tC.north east)+(0,0.16)$)
  node[midway, above=5pt] {\(A=CX\)};

\draw[brace]
  ($(bY.south east)+(0,-0.16)$) -- ($(bC.south west)+(0,-0.16)$)
  node[midway, below=5pt] {\(B=CY\)};

\draw[guide] ($(tC.north west)+(0,0.45)$) -- ($(bC.south west)+(0,-0.45)$);
\draw[guide] ($(tC.north east)+(0,0.45)$) -- ($(bC.south east)+(0,-0.45)$);

\node[fill=white, inner sep=1.5pt] at ($(tC.south)!0.5!(bC.north)$) {\(C=A\cap B\)};

\end{tikzpicture}
\caption{In the factorization partition \(\pi\), the
shaded block is \(A=CX\). In the factorization partition \(\pi'\), the
shaded block is \(B=CY\). Their intersection \(C=A\cap B\) is a block
of the common refinement \(\pi\wedge\pi'\).
}
\label{fig:common-refinement}
\end{figure}
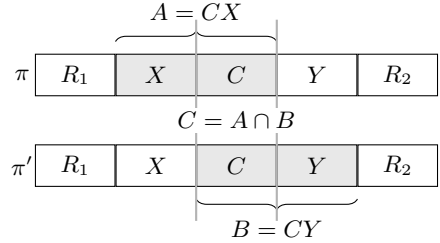

The exact sectors \(\mathscr P_\Lambda=
\{\proj{\psi}:\Lambda_\psi=\Lambda\}\) need not be closed, since a
sequence of states with exact architecture \(\Lambda\) may converge to
a state that factorizes further and therefore has an exact diagram
\(\Gamma\prec\Lambda\). For example, the two-qubit states
\begin{align}
\ket{\psi_\varepsilon}
=
\frac{\ket{00}+\varepsilon\ket{11}}
{\sqrt{1+\varepsilon^2}},
\qquad
\varepsilon>0,
\end{align}
have exact diagram \((2)\) for every \(\varepsilon>0\), but converge as
\(\varepsilon\to0\) to the product state \(\ket{00}\), whose exact
diagram is \((1,1)\).

It is therefore useful to enlarge each exact sector to the set of all
pure states compatible with the corresponding factorization. For
\(\Lambda=(\lambda_1,\ldots,\lambda_{h(\Lambda)})\in\cY_N\), let
\(\mathscr S_\Lambda\) denote the pure states that factorize according
to some partition of the labeled parties into blocks of sizes
\(\lambda_1,\ldots,\lambda_{h(\Lambda)}\). Unlike \(\mathscr P_\Lambda\), the set \(\mathscr S_\Lambda\) also
contains states that factorize further. By Lemma~\ref{lma:1}, every
pure state has a unique exact diagram, and hence
\begin{align}
\mathscr S_\Lambda
=
\bigcup_{\Gamma\preceq\Lambda}\mathscr P_\Gamma.
\label{eq:SM-compatible-sector}
\end{align}

\begin{lemma}[Compact compatible sectors]
\label{lma:SM-compatible-sector-compact}
For every \(\Lambda\in\cY_N\), the set
\(\mathscr S_\Lambda\) is compact.
\end{lemma}

\begin{proof}
For any fixed partition of the parties with shape \(\Lambda\), the
normalized state of each subsystem ranges over a compact set, since
the Hilbert spaces are finite dimensional. The corresponding product
states are obtained by a continuous tensor-product map and therefore
also form a compact set. Since only finitely many labeled partitions
have shape \(\Lambda\), their union \(\mathscr S_\Lambda\) is compact.
\end{proof}
Consequently, a convergent
sequence of pure states with exact diagram \(\Lambda\) can only
converge to a state with exact diagram
\(\Gamma\preceq\Lambda\).

The same compactness extends to arbitrary down-sets. If
\(\down\in\Down(\cY_N)\), then
\begin{align}
\bigcup_{\Gamma\in\down}\mathscr P_\Gamma
=
\bigcup_{\Lambda\in\down}\mathscr S_\Lambda.
\label{eq:SM-finite-union-product-sets}
\end{align}
Indeed, if \(\Lambda\in\down\), every pure state in
\(\mathscr S_\Lambda\) has an exact diagram
\(\Gamma\preceq\Lambda\), and hence \(\Gamma\in\down\). Conversely,
every pure state whose exact diagram is \(\Gamma\in\down\) belongs to
\(\mathscr S_\Gamma\). Since \(\cY_N\) is finite, the right-hand side
is a finite union of compact sets and is therefore compact.

\section{Feasible and irreducible architecture profiles}

We first record the relation between feasible profiles and
pure-state decompositions. Every finite pure-state decomposition
\begin{align}
\rho=\sum_\alpha p_\alpha\proj{\psi_\alpha}
\end{align}
induces the exact-architecture profile
\begin{align}
q_\Lambda
=
\sum_{\alpha:\Lambda_{\psi_\alpha}=\Lambda}p_\alpha.
\label{eq:SM-profile-from-decomp}
\end{align}
Grouping all components with the same exact diagram gives
\begin{align}
\rho
&=
\sum_{\Lambda\in\cY_N}q_\Lambda\rho_\Lambda,
\\
\rho_\Lambda
&=
\frac{1}{q_\Lambda}
\sum_{\alpha:\Lambda_{\psi_\alpha}=\Lambda}
p_\alpha\proj{\psi_\alpha}
\in\mathfrak E_\Lambda,
\label{eq:SM-grouping}
\end{align}
where sectors with \(q_\Lambda=0\) are omitted. Hence
\(q\in\cQ(\rho)\).

Conversely, if \(q\in\cQ(\rho)\), then
\begin{align}
\rho
=
\sum_{\Lambda\in\cY_N}q_\Lambda\rho_\Lambda,
\qquad
\rho_\Lambda\in\mathfrak E_\Lambda.
\label{eq:SM-feasible-profile}
\end{align}
Since
\(\mathfrak E_\Lambda=\conv(\mathscr P_\Lambda)\), every sector state
admits a finite decomposition
\begin{align}
\rho_\Lambda
=
\sum_\beta p_{\beta|\Lambda}
\proj{\psi_{\beta|\Lambda}},
\qquad
\Lambda_{\psi_{\beta|\Lambda}}=\Lambda.
\label{eq:SM-sector-decomp}
\end{align}
Substitution into Eq.~\eqref{eq:SM-feasible-profile} gives a
pure-state decomposition of \(\rho\) whose exact-architecture profile
is \(q\). Thus, \(\cQ(\rho)\) is precisely the set of profiles induced
by pure-state decompositions of \(\rho\).

We next show that every feasible profile dominates an irreducible one.

\begin{theorem}[Pareto domination]
\label{thm:SM-Pareto-domination}
For every \(q\in\cQ(\rho)\), there exists
\(r\in\cP(\rho)\) such that \(r\preceq q\).
\end{theorem}

\begin{proof}
Since the exact sectors need not be closed, existence of a minimal
feasible profile is not immediate; we prove it by minimizing over the
closure and showing that the minimizer is feasible. Let \(q\in\cQ(\rho)\) and consider the set
\begin{align}
K_q
=
\{s\in\cQ(\rho):s\preceq q\},
\end{align}
which is nonempty because $q\in K_q$. Its closure \(\overline{K_q}\) in the finite-dimensional profile
simplex \(\Prob(\cY_N)\) is compact. On this closure, define
\begin{align}
\Phi(s)
=
\sum_{\up\in\Up(\cY_N)\setminus\{\varnothing\}}
\sum_{\Lambda\in\up}s_\Lambda.
\label{eq:SM-Pareto-potential}
\end{align}
This is a continuous function. Moreover,
\(s\prec t\) implies \(\Phi(s)<\Phi(t)\), since every up-set weight of
\(s\) is no larger than the corresponding weight of \(t\), and at
least one is strictly smaller. Hence \(\Phi\) attains its minimum on
\(\overline{K_q}\) at some profile \(x\).

By the definition of the closure, there is a sequence
\(q^{(n)}\in K_q\) converging to \(x\). Since
\(K_q\subseteq\cQ(\rho)\), every \(q^{(n)}\) is realized by a
decomposition of the same state \(\rho\). Let \(d=\dim\cH_N\).
Carathéodory's theorem applied within each exact sector gives at most
\(d^2\) pure states per sector. After padding with zero-weight terms,
we may therefore write
\begin{align}
\rho
&=
\sum_{\Lambda\in\cY_N}
\sum_{j=1}^{d^2}
p_{\Lambda j}^{(n)}\Pi_{\Lambda j}^{(n)},
\label{eq:SM-uniform-decompositions}
\qquad
q_\Lambda^{(n)}
=
\sum_{j=1}^{d^2}p_{\Lambda j}^{(n)},
\end{align}
where \(\Pi_{\Lambda j}^{(n)}\in\mathscr S_\Lambda\), and every
positive-weight projector belongs to
\(\mathscr P_\Lambda\subseteq\mathscr S_\Lambda\).

There are only finitely many indexed coefficients and projectors.
Since \([0,1]\) and every \(\mathscr S_\Lambda\) are compact,
successively extracting subsequences gives a common subsequence along
which
\begin{align}
p_{\Lambda j}^{(n)}
\longrightarrow p_{\Lambda j},
\qquad
\Pi_{\Lambda j}^{(n)}
\longrightarrow\Pi_{\Lambda j}\in\mathscr S_\Lambda
\end{align}
for every \((\Lambda,j)\). Relabeling this subsequence by \(n\), we
still have \(q^{(n)}\to x\).
Because the sums contain a fixed finite number of terms, their limits
can be taken term by term:
\begin{align}
\rho
&=
\lim_{n\to\infty}
\sum_{\Lambda,j}
p_{\Lambda j}^{(n)}\Pi_{\Lambda j}^{(n)}
=
\sum_{\Lambda,j}
p_{\Lambda j}\Pi_{\Lambda j},
\label{eq:SM-limiting-decomposition}
\qquad
x_\Lambda
=
\sum_{j=1}^{d^2}p_{\Lambda j}.
\end{align}

Let \(r\) be the actual exact-architecture profile of the limiting
decomposition. Since a limiting projector originally labeled by
\(\Lambda\) lies in \(\mathscr S_\Lambda\), its exact diagram is some
\(\Gamma\preceq\Lambda\). Reassigning its weight from \(\Lambda\) to
\(\Gamma\) cannot increase the weight of any up-set, and hence
\(r\preceq x\). Moreover, taking \(n\to\infty\) in the up-set inequalities defining
\(q^{(n)}\preceq q\) gives \(x\preceq q\), since every up-set weight is
a finite sum of profile components. Since
\(r\in\cQ(\rho)\) and \(r\preceq x\preceq q\), we have \(r\in K_q\).

Since \(r\preceq x\), inequality \(r\neq x\) would imply
\(r\prec x\), and hence \(\Phi(r)<\Phi(x)\), contradicting the
minimality of \(x\) on \(\overline{K_q}\). Therefore \(r=x\), so
\(x\in\cQ(\rho)\). If \(x\notin\cP(\rho)\), there would exist
\(s\in\cQ(\rho)\) with \(s\prec x\). Since \(x\preceq q\), this gives
\(s\in K_q\) and \(\Phi(s)<\Phi(x)\), again a contradiction. Hence
\(x\in\cP(\rho)\) and \(x\preceq q\), as claimed.
\end{proof}

Since \(\cQ(\rho)\) is nonempty, the theorem also establishes that
\(\cP(\rho)\) is nonempty.

\section{Scalar projections and architectural weights}
\label{sec:mixed-state-constructions}

\subsection{Scalar monotones, convex roofs and depths}
Let
\begin{align}
f:\cY_N\longrightarrow\mathbb R_{\geq0}
\end{align}
be an architectural monotone, that is,
\begin{align}
\Lambda\preceq\Gamma
\quad\Longrightarrow\quad
f(\Lambda)\leq f(\Gamma).
\label{eq:SM-f-monotone}
\end{align}
For a pure state \(\ket{\psi}\), the unique exact diagram
\(\Lambda_\psi\) therefore assigns the scalar value
\(f(\Lambda_\psi)\). Different diagrams may nevertheless have the same
value of \(f\), so this scalar description is generally a projection of
the full architecture rather than a complete characterization of it.

For mixed states, the absence of a preferred pure-state decomposition
leads to several inequivalent extensions. We first consider the
formation-type convex-roof construction
\begin{align}
F_f(\rho)
=
\inf_{
\rho=\sum_\alpha p_\alpha\proj{\psi_\alpha}}
\sum_\alpha p_\alpha f(\Lambda_{\psi_\alpha}),
\label{eq:SM-convex-roof}
\end{align}
which quantifies the smallest average \(f\)-value required in any
formation of \(\rho\). A second possibility is the depth-type
construction
\begin{align}
D_f(\rho)
=
\inf_{
\rho=\sum_\alpha p_\alpha\proj{\psi_\alpha}}
\max_{\alpha:p_\alpha>0}
f(\Lambda_{\psi_\alpha}),
\label{eq:SM-depth-pure-decomp}
\end{align}
which asks for the smallest value \(s\) such that \(\rho\) can be
formed without using any pure component with \(f(\Lambda_\psi)>s\).
These constructions recover, for suitable choices of \(f\),
formation- and depth-type multipartite entanglement quantifiers
\cite{SzalayQuantum2025}.

By the profile--decomposition equivalence established in the previous
section, these optimizations can be written directly over feasible
profiles.
For the formation quantity,
\(\sum_\alpha p_\alpha f(\Lambda_{\psi_\alpha})
=
\sum_{\Lambda\in\cY_N}
q_\Lambda f(\Lambda)
\),
and hence
\begin{align}
F_f(\rho)
=
\inf_{q\in\cQ(\rho)}
\sum_{\Lambda\in\cY_N}
q_\Lambda f(\Lambda).
\label{eq:SM-formation-Q}
\end{align}
Similarly, \(\max_{\alpha:p_\alpha>0}f(\Lambda_{\psi_\alpha})
=
\max_{\Lambda:q_\Lambda>0}f(\Lambda)\), which gives
\begin{align}
D_f(\rho)
=
\inf_{q\in\cQ(\rho)}
\max_{\Lambda:q_\Lambda>0}f(\Lambda).
\label{eq:SM-depth-Q}
\end{align}

By Theorem~\ref{thm:SM-Pareto-domination}, for every
\(q\in\cQ(\rho)\) there exists
\(r\in\cP(\rho)\) such that
   \(r\preceq q\). To see how this affects the formation cost, let \(
t_1<t_2<\cdots<t_L
\)
be the distinct values taken by \(f\) on \(\cY_N\), and define \(\up_k
=
\{\Lambda\in\cY_N:f(\Lambda)\geq t_k\}\).
Monotonicity of \(f\) implies that every \(\up_k\in\Up(\cY_N)\) is indeed an up-set.
For any profile \(q\),
\begin{align}
\sum_{\Lambda}q_\Lambda f(\Lambda)
=
t_1+
\sum_{k=2}^{L}
(t_k-t_{k-1})
\sum_{\Lambda\in\up_k}q_\Lambda.
\label{eq:SM-layer-cake}
\end{align}
Consequently, by definition of the profile order, \(r\preceq q\) implies that
\begin{align}
\sum_\Lambda r_\Lambda f(\Lambda)
\leq
\sum_\Lambda q_\Lambda f(\Lambda).
\label{eq:SM-formation-order-monotonicity}
\end{align}

The largest occupied \(f\)-value is also monotone under the profile
order. Indeed, suppose that \(\max_{\Lambda:r_\Lambda>0}f(\Lambda)
>
\max_{\Lambda:q_\Lambda>0}f(\Lambda)\). Choosing a threshold \(t_k\) strictly above the maximum occupied by
\(q\), but not above the maximum occupied by \(r\), would give \(\sum_{\Lambda\in\up_k}r_\Lambda>0\), and \(\sum_{\Lambda\in\up_k}q_\Lambda=0\), contradicting \(r\preceq q\). Therefore,
\begin{align}
\max_{\Lambda:r_\Lambda>0}f(\Lambda)
\leq
\max_{\Lambda:q_\Lambda>0}f(\Lambda).
\label{eq:SM-depth-order-monotonicity}
\end{align}

Every feasible profile can thus be replaced by a Pareto-minimal profile
with no larger formation or depth cost. Equations
\eqref{eq:SM-formation-Q} and \eqref{eq:SM-depth-Q} consequently become
\begin{align}
F_f(\rho)
&=
\inf_{q\in\cP(\rho)}
\sum_{\Lambda\in\cY_N}
q_\Lambda f(\Lambda),
\label{eq:SM-formation-P}
\\
D_f(\rho)
&=
\inf_{q\in\cP(\rho)}
\max_{\Lambda:q_\Lambda>0}f(\Lambda).
\label{eq:SM-depth-P}
\end{align}
Thus the Pareto frontier contains all formation- and depth-type scalar
quantifiers as monotone projections of the full profile geometry.

\subsection{Down-set classes and depth thresholds}

The depth quantity \(D_f\) answers a membership question: it identifies
the smallest threshold \(s\) for which the state can be formed using
only architectures satisfying \(f(\Lambda)\leq s\). To make this
connection precise, and to prepare the quantitative tail weights used
in the main text, we first associate a convex state class with every
architectural down-set.

Let \(\down\in\Down(\cY_N)\) be nonempty and define
\begin{align}
\mathcal F_\down
=
\conv\left(
\bigcup_{\Lambda\in\down}\mathscr P_\Lambda
\right).
\label{eq:SM-free-class}
\end{align}
Thus, \(\mathcal F_\down\) consists exactly of the states that admit a
pure-state decomposition whose exact diagrams all belong to
\(\down\). We restrict to down-sets because they represent architectural free
classes: whenever an architecture is allowed, every weaker refinement
is allowed as well. Lemma~\ref{lma:SM-compatible-sector-compact} then
has the following mixed-state consequence.

\begin{corollary}[Compact down-set classes]
\label{cor:SM-free-class-compact}
In a finite-dimensional Hilbert space, \(\mathcal F_\down\) is convex
and compact.
\end{corollary}

\begin{proof}
Convexity follows directly from the definition. By
Eq.~\eqref{eq:SM-finite-union-product-sets}, the pure-state union
defining \(\mathcal F_\down\) is compact. Its convex hull is compact in
the finite-dimensional state space.
\end{proof}

In particular, \(\mathcal F_\down\) is closed. This justifies the
implication from vanishing architectural weight to membership in
\(\mathcal F_\down\) used in the main text.

The familiar scalar depth classes arise from our general framework of down-sets as the special case in which the allowed architectures are selected by a threshold on a monotone \(f\):
\begin{align}
\down_{f,s}
=
\{\Lambda\in\cY_N:f(\Lambda)\leq s\},
\qquad s\in f(\cY_N).
\label{eq:SM-f-threshold-downset}
\end{align}
Monotonicity of \(f\) ensures that \(\down_{f,s}\) is a down-set.
Membership in the corresponding class is equivalent to the existence
of a decomposition containing no component above the threshold:
\begin{align}
\rho\in\mathcal F_{\down_{f,s}}
\quad\Longleftrightarrow\quad
\rho
\text{ admits a decomposition with }
f(\Lambda_{\psi_\alpha})\leq s
\text{ for all }\alpha.
\label{eq:SM-threshold-membership}
\end{align}
Consequently,
\begin{align}
D_f(\rho)
=
\min\left\{
s\in f(\cY_N):
\rho\in\mathcal F_{\down_{f,s}}
\right\}.
\label{eq:SM-depth-membership}
\end{align}
The minimum exists because \(f(\cY_N)\) is finite. Thus, a
depth-type quantifier compresses a nested family of architectural
membership questions into a single number.

\subsection{Architectural tail weights}

Depth distinguishes whether all components outside a chosen class can be
eliminated, but it does not quantify how much of them is unavoidable.
This motivates the architectural tail weight introduced for a nonempty down-set \(\down\) as
\begin{align}
w_\down(\rho)
=
\inf_{q\in\cQ(\rho)}
\ell_\down(q),
\label{eq:SM-tail-weight}
\end{align}
where
\begin{align}
\ell_\down(q)
=
\sum_{\Lambda\notin\down}q_\Lambda.
\label{eq:SM-profile-leakage}
\end{align}
It is the smallest total formation weight that must be assigned outside
the architectural class \(\down\).

For a scalar-generated down-set \(\down_{f,s}\), this becomes the
minimum fraction of every formation lying above the scalar threshold
\(s\). In particular,
\begin{align}
w_{\down_{f,s}}(\rho)=0
\quad\Longleftrightarrow\quad
D_f(\rho)\leq s.
\label{eq:SM-tail-depth-relation}
\end{align}
Hence \(w_{\down_{f,s}}\) refines the corresponding depth criterion from
a binary membership statement to a quantitative one.

The tail weight can also be viewed as a convex-roof construction. Define
the indicator cost
\begin{align}
\chi_\down(\Lambda)
=
\begin{cases}
0, & \Lambda\in\down,\\
1, & \Lambda\notin\down.
\end{cases}
\label{eq:SM-downset-indicator}
\end{align}
Then
\begin{align}
w_\down(\rho)
=
\inf_{
\rho=\sum_\alpha p_\alpha\proj{\psi_\alpha}}
\sum_\alpha
p_\alpha
\chi_\down(\Lambda_{\psi_\alpha})=\inf_{q\in\cP(\rho)}
\sum_{\Lambda\in\cY_N}
q_\Lambda \chi_\down(\Lambda).
\label{eq:SM-tail-indicator-roof}
\end{align}
Indeed, the average indicator cost of a decomposition is exactly the
total weight assigned to pure components whose exact diagrams lie
outside \(\down\). Thus the architectural tail weight is the
formation-type convex roof generated by the binary cost
\(\chi_\down\).

Conversely, general scalar formation costs can be reconstructed from a nested
family of such indicator costs. Let \(t_1<t_2<\cdots<t_L\) be the distinct values taken by \(f\) on \(\cY_N\). For any profile
\(q\), we obtain by applying the $\inf$ on both sides of Eq.~\eqref{eq:SM-layer-cake} that
\begin{align}
F_f(\rho)
=
t_1+
\inf_{q\in\cQ(\rho)}
\sum_{k=1}^{L-1}
(t_{k+1}-t_k)
\ell_{\down_{f,t_k}}(q).
\label{eq:SM-formation-from-tail-leakages}
\end{align}
Thus a scalar formation measure is obtained by optimizing a weighted sum
of leakages across the nested threshold down-sets generated by \(f\).

In general,
\begin{align}
F_f(\rho)
\geq
t_1+
\sum_{k=1}^{L-1}
(t_{k+1}-t_k)
w_{\down_{f,t_k}}(\rho),
\label{eq:SM-formation-tail-lower-bound}
\end{align}
because the infimum of a sum is not smaller than the sum of the separate
infima. Equality need not hold: the profiles minimizing the different
tail weights may be incompatible. This is precisely the information
retained by the joint profile geometry and lost by optimizing each
threshold independently.

\subsection{Resource-weight interpretation}

We now show that the tail weight is precisely the resource
weight associated with the convex class \(\mathcal F_\down\), which was defined in Ref.~\cite{DucuaraPRL2020} and generalizes the
best separable approximation~\cite{LewensteinPRL1998} as
\begin{align}
W_{\mathcal F_\down}(\rho)
=
\inf_{\substack{
p\in[0,1],\\ \sigma_\down\in\mathcal F_\down,\\
\tau\in\mathcal S(\cH_N)
}}
\left\{
p:
\rho=(1-p)\sigma_\down+p\tau
\right\}.
\label{eq:SM-standard-resource-weight}
\end{align}

\begin{proposition}[Resource-weight equivalence]
For every nonempty down-set \(\down\),
\[
w_\down(\rho)=W_{\mathcal F_\down}(\rho).
\]
\end{proposition}
\begin{proof}
Take any feasible profile \(q\in\cQ(\rho)\) and set
\(p=\ell_\down(q)\). For \(0<p<1\), define
\begin{align}
\sigma_\down
=
\frac{1}{1-p}
\sum_{\Lambda\in\down}
q_\Lambda\rho_\Lambda,
\label{eq:SM-inside-state}
\end{align}
and
\begin{align}
\tau
=
\frac{1}{p}
\sum_{\Lambda\notin\down}
q_\Lambda\rho_\Lambda.
\label{eq:SM-outside-state}
\end{align}
Then
\(\sigma_\down\in\mathcal F_\down\) and
\(\rho=(1-p)\sigma_\down+p\tau\). The cases \(p=0\) and \(p=1\) follow by omitting the absent term. Thus,
for every \(q\in\cQ(\rho)\), the above construction yields an
admissible resource decomposition with
\(p=\ell_\down(q)\), and, therefore \(W_{\mathcal F_\down}(\rho)
\leq
\ell_\down(q)\). Since this construction is valid for every
\(q\in\cQ(\rho)\), it follows that
\begin{align}
W_{\mathcal F_\down}(\rho)
\leq
w_\down(\rho).
\label{eq:SM-resource-weight-first-direction}
\end{align}

Conversely, consider any decomposition \(\rho=(1-p)\sigma_\down+p\tau\),
with \(\sigma_\down\in\mathcal F_\down\). Choose pure-state
decompositions of \(\sigma_\down\) and \(\tau\). All weight arising from
\((1-p)\sigma_\down\) lies inside \(\down\), whereas at most the weight
\(p\) arising from \(\tau\) can lie outside it. The induced exact profile therefore satisfies
\(\ell_\down(q)\leq p\). By the definition of \(w_\down\), this implies \(
w_\down(\rho)
\leq
p
\) for every admissible decomposition
\(\rho=(1-p)\sigma_\down+p\tau\), with
\(\sigma_\down\in\mathcal F_\down\). Taking the infimum over all such
resource decompositions therefore gives
\begin{align}
w_\down(\rho)
\leq
W_{\mathcal F_\down}(\rho).
\label{eq:SM-resource-weight-second-direction}
\end{align}
Combining this with
Eq.~\eqref{eq:SM-resource-weight-first-direction}, we obtain
\begin{align}
w_\down(\rho)
=
W_{\mathcal F_\down}(\rho).
\label{eq:SM-resource-weight-equivalence}
\end{align}

Because \(\mathcal F_\down\) and the full state space are compact, the
infimum is attained. Therefore \(w_\down(\rho)\) is the smallest
\(p\in[0,1]\) such that \(\rho=(1-p)\sigma_\down+p\tau\), with \(\sigma_\down\in\mathcal F_\down\) and
\(\tau\in\mathcal S(\cH_N)\).
\end{proof}

\subsection{General properties of architectural weights}

The equivalent formulations above make the main properties of
\(w_\down\) immediate. First, the convex-roof representation
\eqref{eq:SM-tail-indicator-roof} shows that \(w_\down\) is convex in
\(\rho\) and reduces on pure states to the indicator function
\eqref{eq:SM-downset-indicator}.

The weight is faithful:
\begin{align}
w_\down(\rho)=0
\quad\Longleftrightarrow\quad
\rho\in\mathcal F_\down.
\label{eq:SM-weight-faithful}
\end{align}
The admissible set in the resource-weight formulation
\eqref{eq:SM-standard-resource-weight} is compact, since
\([0,1]\), \(\mathcal F_\down\), and the full state space are compact
and the decomposition constraint is closed. Hence the infimum is
attained. Therefore \(w_\down(\rho)=0\) is equivalent to the existence
of an admissible decomposition with \(p=0\), that is,
\(\rho\in\mathcal F_\down\).

As a resource weight for the convex free set
\(\mathcal F_\down\), \(w_\down\) is nonincreasing under quantum
channels \(\Phi\) satisfying
\(\Phi(\mathcal F_\down)\subseteq\mathcal F_\down\)
\cite{DucuaraPRL2020,ChitambarRMP2019}.
Moreover, if \(\down_1\subseteq\down_2\), then
\(\ell_{\down_1}(q)\geq\ell_{\down_2}(q)\) for every feasible profile
\(q\), and therefore \(
w_{\down_1}(\rho)
\geq
w_{\down_2}(\rho)\).

For scalar-generated classes, Eq.~\eqref{eq:SM-weight-faithful}
together with Eq.~\eqref{eq:SM-depth-membership} gives
\begin{align}
w_{\down_{f,s}}(\rho)=0
\quad\Longleftrightarrow\quad
D_f(\rho)\leq s.
  \label{eq:SM-tail-depth-faithful}
 \end{align}
Thus the tail weight quantitatively refines the corresponding depth
membership criterion.

\section{A four-qubit state with a continuous frontier}
Here we show that the irreducible architecture of the state $\rho_{\rm P}$ is not a single profile, but a continuous Pareto frontier. Let us first recall some definitions from the main text:
\begin{align}
\ket{A}=\ket{\Phi^+}_{12}\otimes\ket{\Phi^+}_{34},
\qquad
\ket{B}=\ket{\mathrm{GHZ}^+}_{123}\otimes\ket{-}_4.
\label{eq:A-B-def}
\end{align}
where
\begin{align}
&\ket{\Phi^+}=\frac{\ket{00}+\ket{11}}{\sqrt2},
\qquad
\ket{\mathrm{GHZ}^+}=\frac{\ket{000}+\ket{111}}{\sqrt2},
\nonumber\\
&\ket{-}=\frac{\ket0-\ket1}{\sqrt2}.
\end{align}
Then \(\ket{A}\) has exact type \((2,2)\), \(\ket{B}\) has exact type \((3,1)\), and \(\langle A|B\rangle=0\). Define
\begin{align}
\rho_{\rm P}
=
\frac{1}{2}\proj{A}
+
\frac{1}{2}\proj{B}
+
\frac{1}{4}\left(\ket{A}\!\bra{B}+\ket{B}\!\bra{A}\right).
  \label{eq:SM-rhoP}
 \end{align}
In the basis \(\{\ket{A},\ket{B}\}\),
\begin{align}
\rho_{\rm P}
=
\begin{pmatrix}
1/2&1/4\\
1/4&1/2
\end{pmatrix},
\end{align}
so \(\rho_{\rm P}\) is a rank-two mixed state with eigenvalues \(3/4\) and \(1/4\).

\begin{lemma}[Exact architecture sectors of $\rho_{\rm P}$]\label{lma:profilesrhoP} The two-dimensional support of \(\rho_{\rm P}\) contains only three exact Young-diagram types. Every nonzero vector in
\begin{align}
S=\operatorname{span}\{\ket{A},\ket{B}\}
\end{align}
has the form
\begin{align}
\ket{\psi(a,b)}=a\ket{A}+b\ket{B},
\end{align}
and
\begin{align}
\Lambda_{\psi(a,b)}
=
\begin{cases}
(2,2), & b=0,\\
(3,1), & a=0,\\
(4), & ab\ne0.
\end{cases}
\label{eq:support-classification-main}
\end{align}
\end{lemma}
Since every pure state in any decomposition of \(\rho_{\rm P}\) must lie in \(S\), Eq.~\eqref{eq:support-classification-main} completely determines the possible exact sectors: only \((2,2)\), \((3,1)\), and \((4)\) can appear.

\begin{proof}
The two vectors are
\begin{align}
2\ket A&=\ket{0000}+\ket{0011}+\ket{1100}+\ket{1111},\\
2\ket B&=\ket{0000}-\ket{0001}+\ket{1110}-\ket{1111}.
\end{align}
Thus a general linear combination of both is
\begin{align}
\ket{\psi(a,b)}&=a\ket A+b\ket B\notag\\
&=
\frac{1}{2}\Big[
(a+b)\ket{0000}
-b\ket{0001}
+a\ket{0011}
+a\ket{1100}
+b\ket{1110}
+(a-b)\ket{1111}
\Big].
\end{align}

The exact type of a pure four-qubit state is fixed by the bipartitions
across which it factorizes. For a given cut \(X|\bar X\), we write
\(\ket{\psi(a,b)}\) as a bipartite state,
\begin{align}
\ket{\psi(a,b)}
=
\sum_{\mu,\nu} M^{X|\bar X}_{\mu\nu}\ket{\mu}_X\otimes\ket{\nu}_{\bar X}.
\end{align}
The Schmidt rank across \(X|\bar X\) is the rank of the coefficient
matrix \(M^{X|\bar X}\). Hence the state is a product across this cut
if and only if \(M^{X|\bar X}\) has rank one. We use the elementary
rank criterion that a matrix has rank one only if all its \(2\times2\)
minors vanish; conversely, a single nonzero \(2\times2\) minor proves
Schmidt rank at least two.

We first consider the cut \(12|34\). In the product bases
\(\{\ket{00},\ket{01},\ket{10},\ket{11}\}_{12}\) and
\(\{\ket{00},\ket{01},\ket{10},\ket{11}\}_{34}\), the coefficient
matrix is
\begin{align}
M_{12|34}(a,b)
=
\frac{1}{2}
\begin{pmatrix}
 a+b & -b & 0 & a\\
 0&0&0&0\\
 0&0&0&0\\
 a&0&b&a-b
\end{pmatrix}.
\end{align}
The minor formed from rows \(\ket{00},\ket{11}\) and columns
\(\ket{01},\ket{10}\) has determinant \(-b^2/4\). Therefore Schmidt
rank one across \(12|34\) requires \(b=0\). Conversely, if \(b=0\), the
vector is proportional to
\(\ket A=\ket{\Phi^+}_{12}\otimes\ket{\Phi^+}_{34}\), which is a
product across \(12|34\) but not within the pairs. This gives exact
type \((2,2)\).

For the cut \(123|4\), in the product bases
\(\{\ket{000},\ket{001},\ket{010},\ket{011},
\ket{100},\ket{101},\ket{110},\ket{111}\}_{123}\) and
\(\{\ket0,\ket1\}_4\), the coefficient matrix is
\begin{align}
M_{123|4}(a,b)
=
\frac{1}{2}
\begin{pmatrix}
a+b & -b\\
0 & a\\
0 & 0\\
0 & 0\\
0 & 0\\
0 & 0\\
a & 0\\
b & a-b
\end{pmatrix}.
\end{align}
The minor formed from rows \(\ket{001},\ket{110}\) and columns
\(\ket0,\ket1\) has determinant \(-a^2/4\). Therefore Schmidt rank one
across \(123|4\) requires \(a=0\). Conversely, if \(a=0\), the vector is
proportional to
\(\ket B=\ket{\mathrm{GHZ}^+}_{123}\otimes\ket{-}_4\). Since
\(\ket{\mathrm{GHZ}^+}_{123}\) does not factorize across any bipartition
of qubits \(1,2,3\), this gives exact type \((3,1)\).

The remaining bipartitions give no further nonzero product rays. The
following table lists representative \(2\times2\) minors of the
corresponding coefficient matrices. The displayed minors are sufficient
to obtain the condition in the last column.
\begin{center}
\begin{tabular}{c|c|c}
cut & sufficient \(2\times2\) minors & rank-one condition \\
\hline
\(12|34\) & \(-b^2/4\) & \(b=0\) \\
\(123|4\) & \(-a^2/4\) & \(a=0\) \\
\(1|234\), \(2|134\) & \(a^2/4,\ -b^2/4\) & \(a=b=0\) \\
\(3|124\) &
\(a(a+b)/4,\ a(a-b)/4,\ b(a+b)/4,\ b(a-b)/4\)
& \(a=b=0\) \\
\(13|24\), \(14|23\) & \(a^2/4,\ -b^2/4\) & \(a=b=0\)
\end{tabular}
\end{center}
Each entry is understood up to the sign shown for one convenient choice
of rows and columns. Thus the only nonzero vectors in the support that
factorize across any bipartition are the rays \(b=0\) and \(a=0\). Every
vector with \(ab\ne0\) has Schmidt rank larger than one across every
nontrivial bipartition, and hence has no nontrivial tensor-product
factorization. Its exact type is therefore \((4)\). We have proved
\begin{align}
\Lambda_{\psi(a,b)}
=
\begin{cases}
(2,2),& b=0,\\
(3,1),& a=0,\\
(4),& ab\ne0.
\end{cases}
\end{align}
\end{proof}

Using Lemma~\ref{lma:profilesrhoP}, we write a feasible profile of $\rho_{\rm P}$ as
\begin{align}
q=(x,y,z)
=
\bigl(q_{(2,2)},q_{(3,1)},q_{(4)}\bigr),
\qquad
z=1-x-y.
\end{align}
Because the only exact-\((2,2)\) ray in the support is \(\ket{A}\), and the only exact-\((3,1)\) ray is \(\ket{B}\), assigning weights \(x\) and \(y\) to these two sectors leaves the residual operator
\begin{align}
R(x,y)
:=
\rho_{\rm P}-x\proj{A}-y\proj{B}.
\end{align}
The profile \((x,y,z)\) can only be feasible when \(R(x,y)\ge0\). In the basis \(\{\ket{A},\ket{B}\}\),
\begin{align}
R(x,y)
=
\begin{pmatrix}
1/2-x&1/4\\
1/4&1/2-y
\end{pmatrix}.
\end{align}
Therefore
\begin{align}
(x,y,z)\in\cQ(\rho_{\rm P})
\quad\Longleftrightarrow\quad
\left(\frac{1}{2}-x\right)
\left(\frac{1}{2}-y\right)
\ge
\frac{1}{16},
\label{eq:feasible-region}
\end{align}
with \(x,y\ge0\) and \(z=1-x-y\ge0\). Positivity is also sufficient: whenever \(R(x,y)\ge0\), its nonzero eigenvectors have nonzero components along both \(\ket{A}\) and \(\ket{B}\), because the off-diagonal element is fixed and nonzero. By Eq.~\eqref{eq:support-classification-main}, these eigenvectors are exact-\((4)\) states. Hence the residual is generated by exact-\((4)\) states.

In this three-sector support, the profile order has a simple interpretation. A profile is architecturally weaker if it has no more exact-\((4)\) weight and, correspondingly, no less weight in the lower sectors \((2,2)\) and \((3,1)\). Thus, for profiles supported on \((2,2)\), \((3,1)\), and \((4)\),
\begin{align}
r\preceq q
\quad\Longleftrightarrow\quad
x_r\ge x_q,
\qquad
y_r\ge y_q,
\qquad
z_r\le z_q.
\end{align}
These three conditions follow by applying the condition Eq.~(\ref{eq:profile-order}) to the three nontrivial upward-closed sets
\(\{(4)\}\), \(\{(2,2),(4)\}\), and \(\{(3,1),(4)\}\). Together they express that decreasing the weight of the strongest sector \((4)\), while replacing it by weight in the lower incomparable sectors, gives an architecturally weaker profile.

\begin{theorem}[Continuous Pareto frontier]
\label{thm:continuous-frontier}
For the four-qubit state \(\rho_{\rm P}\), the exact-architecture Pareto frontier is the continuous curve
\begin{align}
\cP(\rho_{\rm P})
=
\left\{
(x,y,1-x-y):
\left(\frac{1}{2}-x\right)
\left(\frac{1}{2}-y\right)
=
\frac{1}{16}
\right\},
  \label{eq:SM-frontier-curve}
 \end{align}
with \(0\le x,y\le3/8\). Hence the irreducible exact architecture of \(\rho_{\rm P}\) is not a single profile.
\end{theorem}

\begin{proof}
    Every pure state appearing in any decomposition of \(\rho_{\rm P}\) must lie in the support \(S\). Indeed, if \(\rho=\sum_jp_j\proj{\psi_j}\) and \(\ket\eta\in\ker\rho\), then
\begin{align}
0=\bra\eta\rho\ket\eta=\sum_jp_j|\langle\eta|\psi_j\rangle|^2,
\end{align}
so each \(\ket{\psi_j}\) is orthogonal to \(\ker\rho\), i.e. lies in \(\supp\rho=S\).

A profile with \(x=q_{(2,2)}\), \(y=q_{(3,1)}\), and
\(z=q_{(4)}=1-x-y\) is feasible if and only if
\begin{align}\label{eq:feasibilitysupp}
\left(\frac{1}{2}-x\right)\left(\frac{1}{2}-y\right)\ge\frac{1}{16},
\qquad
x\ge0,\qquad y\ge0,
\qquad z=1-x-y\ge0.
\end{align}
Moreover, we have seen that \((x',y',z')\preceq (x,y,z)\) if
and only if \(z'\le z\), \(x'\ge x\), and \(y'\ge y\).
Interior feasible profiles are dominated, because one can increase
\(x\) or \(y\) slightly while remaining feasible. Boundary profiles are
not dominated, because any feasible \(q'\ne q\) with
\(q'\preceq q\) would require \(x'\ge x\) and \(y'\ge y\),
with at least one inequality strict, which leads to
\begin{align}
    \left(\frac{1}{2}-x'\right)\left(\frac{1}{2}-y'\right)
<
\left(\frac{1}{2}-x\right)\left(\frac{1}{2}-y\right)
=
\frac{1}{16},
\end{align}
violating the feasibility
condition~(\ref{eq:feasibilitysupp}).
Therefore the Pareto frontier is
exactly
\begin{align}
\left(\frac{1}{2}-x\right)\left(\frac{1}{2}-y\right)=\frac{1}{16},
\qquad
0\le x,y\le\frac{3}{8},
\qquad
z=1-x-y.
\end{align}
\end{proof}

The endpoints of the frontier are
\begin{align}
q_A=\left(\frac{3}{8},0,\frac{5}{8}\right),
\qquad
q_B=\left(0,\frac{3}{8},\frac{5}{8}\right).
\end{align}
They are realized by the decompositions
\begin{align}
\rho_{\rm P}
=
\frac{3}{8}\proj{A}
+
\frac{5}{8}\proj{\psi_A},
\qquad
\ket{\psi_A}
=
\frac{\ket{A}+2\ket{B}}{\sqrt5},
\end{align}
and
\begin{align}
\rho_{\rm P}
=
\frac{3}{8}\proj{B}
+
\frac{5}{8}\proj{\psi_B},
\qquad
\ket{\psi_B}
=
\frac{2\ket{A}+\ket{B}}{\sqrt5}.
\end{align}
The states \(\ket{\psi_A}\) and \(\ket{\psi_B}\) are exact \((4)\) by Eq.~\eqref{eq:support-classification-main}. The two endpoints are incomparable because \((2,2)\) and \((3,1)\) are incomparable Young diagrams. Interior points of the frontier contain all three sectors \((2,2)\), \((3,1)\), and \((4)\), with continuously varying weights.

\begin{lemma}[Minimum GME weight]\label{lma:minGME} The Pareto-minimal profile contains an unavoidable exact-\((4)\) contribution of
\begin{align}
q_{(4)}=1-x-y\ge\frac{1}{2},
\end{align}
with equality at the symmetric point \(x=y=1/4\).
\end{lemma}

\begin{proof}
Along the frontier, set \(u=1/2-x\) and \(v=1/2-y\). Then the boundary
condition fixes the product, \(uv=1/16\), while the exact-\((4)\) weight
is \(q_{(4)}=1-x-y=u+v\). For two positive numbers with fixed product,
the sum is minimized when the two numbers are equal, since
\(u+v\ge2\sqrt{uv}\), with equality iff \(u=v\). Hence
\(q_{(4)}\ge2\sqrt{1/16}=1/2\), with equality at \(u=v=1/4\), i.e.
\(x=y=1/4\). This establishes that \(\rho_{\rm P}\) has an
unavoidable exact-\((4)\) formation weight of at least \(1/2\), and this
minimum is attained only by the symmetric Pareto-minimal profile.
\end{proof}

Thus the irreducible architecture of \(\rho_{\rm P}\) cannot be reduced to a unique profile. After all dominated profiles have been removed, one is still left with a continuous Pareto frontier of inequivalent minimal
formations, trading the two lower incomparable architectures \((2,2)\)
and \((3,1)\). Moreover, this trade-off never removes the genuine
four-party component: every irreducible formation requires exact-\((4)\) weight at least \(1/2\).

The four-qubit construction extends to the tunable-coherence family
\begin{align}
\rho_\alpha
=
\frac{1}{2}\proj A+\frac{1}{2}\proj B
+\alpha\ket A\!\bra B+\alpha^*\ket B\!\bra A,
\qquad
|\alpha|\leq\frac{1}{2}.
\label{eq:SM-tunable-coherence}
\end{align}
For \(0<|\alpha|<1/2\), its feasible profiles satisfy
\begin{align}
\left(\frac{1}{2}-x\right)
\left(\frac{1}{2}-y\right)
\geq|\alpha|^2,
\qquad x,y\geq0,
\qquad z=1-x-y\geq0,
\end{align}
and the Pareto frontier is the equality boundary, with
\(0\leq x,y\leq1/2-2|\alpha|^2\). Its endpoints are
\((1/2-2|\alpha|^2,0,1/2+2|\alpha|^2)\) and
\((0,1/2-2|\alpha|^2,1/2+2|\alpha|^2)\), while
\(z_{\min}=2|\alpha|\) is attained at
\(x=y=1/2-|\alpha|\). At \(\alpha=0\), the unique Pareto-minimal
profile is \((1/2,1/2,0)\); at \(|\alpha|=1/2\), the state is pure of
exact type \((4)\), with profile \((0,0,1)\).

\section{One-axis-twisting examples}
\label{sec:SM-OAT}

\subsection{Exact QFI for ideal one-axis twisting}

Figure~2 of the main text uses the exact QFI of the pure OAT state
\begin{align}
\ket{\psi_\tau}
=
U_\tau\ket{+x}^{\otimes N},
\qquad
U_\tau=e^{-i\tau J_z^2}.
\end{align}
For a collective generator \(J_{\bm n}=\bm n\cdot\bm J\), pure-state
QFI satisfies
\(F_Q[\proj{\psi},J_{\bm n}]=4(\Delta J_{\bm n})^2\)
\cite{PezzeRMP2018}. Optimization over \(\bm n\) therefore gives
\begin{align}
F_Q^{\rm opt}(\tau)
=
4\lambda_{\max}\!\left[\Gamma^{(1)}(\tau)\right],
\label{eq:SM-OAT-optimal-QFI}
\end{align}
where
\(\Gamma^{(1)}_{kl}
=\langle\{J_k,J_l\}\rangle/2
-\langle J_k\rangle\langle J_l\rangle\).

   For \(N\ge2\), the required nonzero moments are
\cite{KitagawaPRA1993,BaamaraPRL2021,BaamaraCRPhys2022}
\begin{align}
\langle J_x\rangle
&=
\frac{N}{2}\cos^{N-1}\tau,
\\
\langle J_x^2\rangle
&=
\frac{N(N+1)}{8}
+
\frac{N(N-1)}{8}\cos^{N-2}(2\tau),
\\
\langle J_y^2\rangle
&=
\frac{N(N+1)}{8}
-
\frac{N(N-1)}{8}\cos^{N-2}(2\tau),
\\
\langle J_z^2\rangle
&=
\frac{N}{4},
\\
\frac{1}{2}\langle\{J_y,J_z\}\rangle
&=
\frac{N(N-1)}{4}
\sin\tau\,\cos^{N-2}\tau,
\label{eq:SM-OAT-moments}
\end{align}
while the remaining first moments and symmetrized cross correlations
vanish. For each \(\tau\), Fig.~2 was obtained by inserting
\(F_Q^{\rm opt}(\tau)\) into
Eq.~\eqref{eq:QFI-downset-weight} of the main text. The lower panel
retains only whether the corresponding numerator is positive,
\(F_Q^{\rm opt}(\tau)>s_\down\).

\subsection{Nonlinear-squeezing bounds under collective dephasing}

We now include collective dephasing relevant to OAT realizations in
trapped ions and cold atoms, following Ref.~\cite{BaamaraCRPhys2022}.
Instead of the exact QFI, we use the experimentally accessible
nonlinear-squeezing lower bound~\cite{GessnerPRL2019}, implemented by
measurement after interaction (MAI). A rotation angle \(\theta\) is
encoded as
\(\rho_\theta=e^{-i\theta J_{\bm n}}\rho e^{i\theta J_{\bm n}}\), with
normalized \(\bm n\) restricted to the \(yz\)-plane. The MAI protocol
uses two OAT evolutions, one for state preparation and a second with
inverted interaction sign, before measuring an optimal collective spin observable in the $yz$-plane. This
corresponds to the effective observables
\(\bm X=(U_{-\tau}^\dagger J_yU_{-\tau},
U_{-\tau}^\dagger J_zU_{-\tau})^\top\).
From
\[C_{kl}=-i\langle[X_l,J_k]\rangle\quad\text{and}\quad
\Gamma_{kl}=\langle\{X_k,X_l\}\rangle/2
-\langle X_k\rangle\langle X_l\rangle,\]
evaluated at $\theta=0$, optimization over the rotation axis and measured linear combination
then gives the lower bound
\begin{equation}
F_Q^{\rm opt}(\rho)
\geq\lambda_{\max}\!\left(C\Gamma^{-1}C^\top\right).
\label{eq:SM-NL-QFI-bound}
\end{equation}
The results are shown in Fig.~\ref{fig:SM-OAT-dephasing}.

\subsubsection{Ballistic collective dephasing}

Ballistic dephasing replaces the preparation and measurement OAT
propagators by
\(U_{\rm p}(\delta)=e^{-i\tau(J_z^2+\delta J_z)}\) and
\(U_{\rm m}(\delta)=e^{-i\tau(-J_z^2+\delta J_z)}\), where \(\delta\)
is Gaussian with \(\langle\delta\rangle=0\) and fixed variance
\(\langle\delta^2\rangle\). The same \(\delta\) acts throughout both
stages: only the OAT term changes sign, while the dephasing phase
continues to accumulate and the
average is taken only after the complete protocol. The response matrix
is
\begin{align}
C_{kl}
={}&
\left.\partial_\theta
\int d\delta\,
\frac{e^{-\delta^2/[2\langle\delta^2\rangle]}}
{\sqrt{2\pi\langle\delta^2\rangle}}\,
\bra{+x}^{\otimes N}
U_{\rm p}^\dagger(\delta)e^{i\theta J_k}
U_{\rm m}^\dagger(\delta)J_lU_{\rm m}(\delta)
e^{-i\theta J_k}U_{\rm p}(\delta)
\ket{+x}^{\otimes N}
\right|_{\theta=0}
\nonumber\\
={}&
i\int d\delta\,
\frac{e^{-\delta^2/[2\langle\delta^2\rangle]}}
{\sqrt{2\pi\langle\delta^2\rangle}}\,
\bra{+x}^{\otimes N}
U_{\rm p}^\dagger(\delta)
\bigl[J_k,U_{\rm m}^\dagger(\delta)J_lU_{\rm m}(\delta)\bigr]
U_{\rm p}(\delta)
\ket{+x}^{\otimes N}.
\label{eq:SM-ballistic-response}
\end{align}
At \(\theta=0\), the first moments vanish and the OAT evolutions cancel,
so that
\begin{align}
\Gamma_{kl}
={}&
\frac12\int d\delta\,
\frac{e^{-\delta^2/[2\langle\delta^2\rangle]}}
{\sqrt{2\pi\langle\delta^2\rangle}}\,
\bra{+x}^{\otimes N}
e^{2i\delta\tau J_z}\{J_k,J_l\}
e^{-2i\delta\tau J_z}
\ket{+x}^{\otimes N}.
\label{eq:SM-ballistic-covariance}
\end{align}
This gives
\begin{equation}
\begin{aligned}
C_{\rm bal}
&=\frac N2
\begin{pmatrix}
\dfrac{N-1}{2}
\left(
e^{-9\langle\delta^2\rangle\tau^2/2}
+e^{-\langle\delta^2\rangle\tau^2/2}
\right)
\sin\tau\cos^{N-2}\tau
&
-e^{-\langle\delta^2\rangle\tau^2/2}\cos^{N-1}\tau
\\
e^{-2\langle\delta^2\rangle\tau^2} & 0
\end{pmatrix},
\\[1mm]
\Gamma_{\rm bal}
&=
\begin{pmatrix}
\dfrac{N(N+1)}8
-\dfrac{N(N-1)}8e^{-8\langle\delta^2\rangle\tau^2}
&0
\\
0&\dfrac N4
\end{pmatrix}.
\end{aligned}
\label{eq:SM-ballistic-MAI}
\end{equation}
These results coincide with Eq.~(57) in Ref.~\cite{BaamaraCRPhys2022},
but differ from the expressions given in its Appendix~C, which describe
the case in which the dephasing term also changes sign during the
readout OAT.

\begin{figure}[t]
    \centering
    \includegraphics[width=.49\linewidth]
    {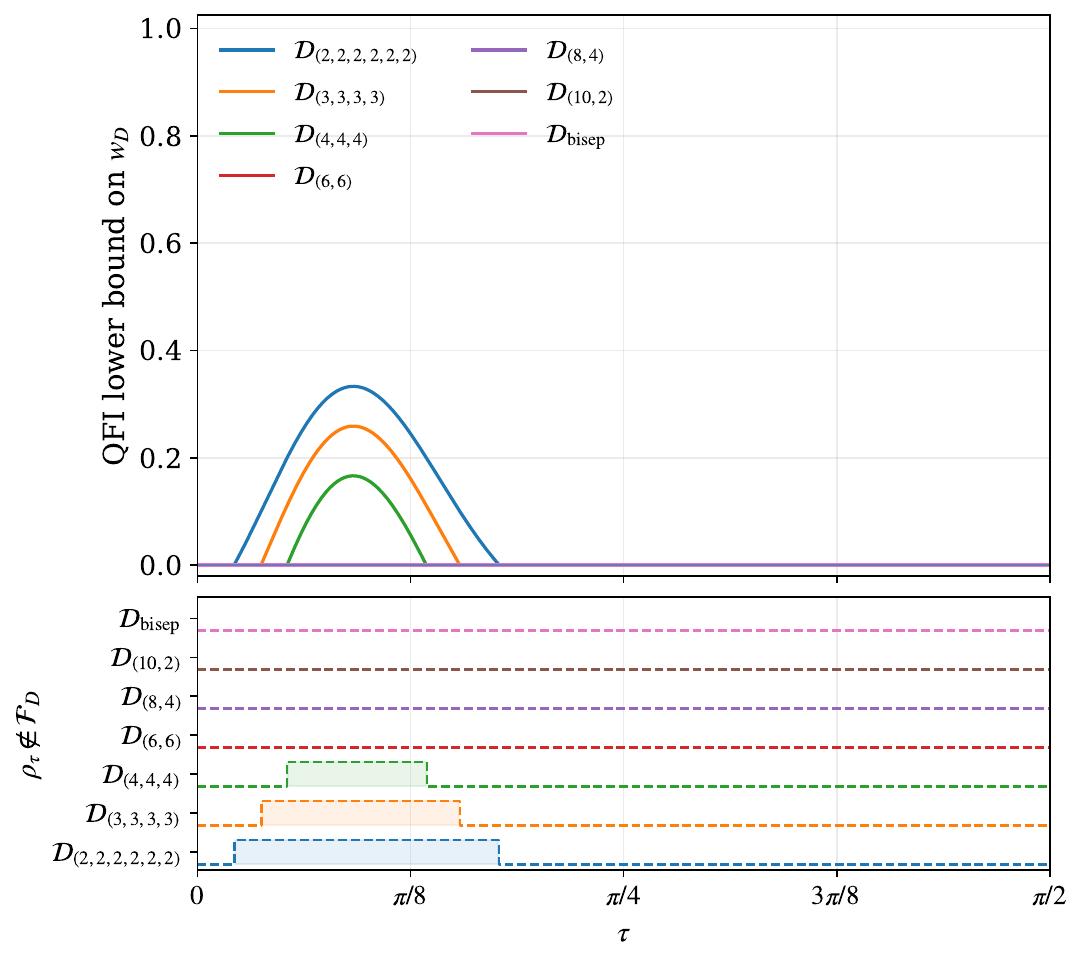}
    \hfill
    \includegraphics[width=.49\linewidth]
    {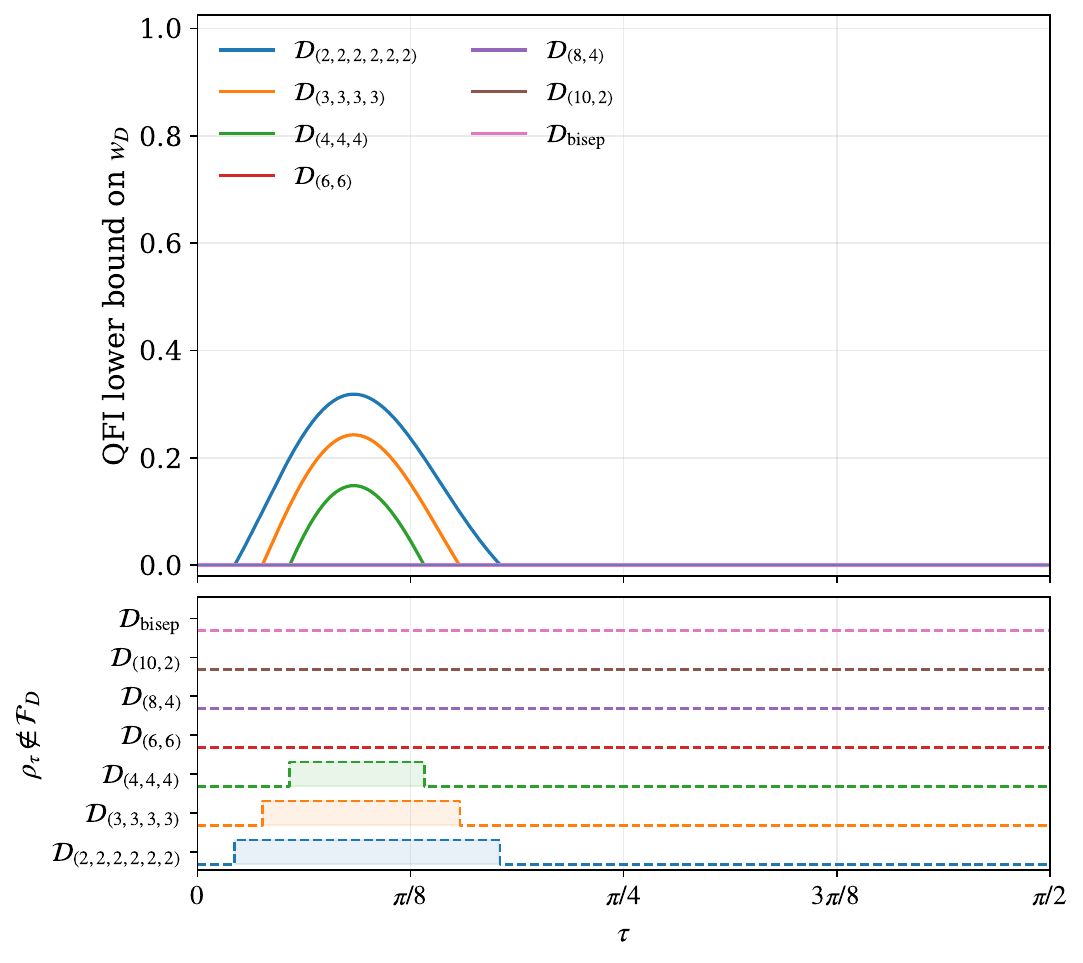}
    \caption{Architectural-weight bounds for \(N=12\) OAT states under
    collective dephasing, obtained from the MAI lower bound on the QFI:
    (left panel) ballistic dephasing with
    \(\langle\delta^2\rangle=0.01\);
    (right panel) diffusive dephasing with
    \(\epsilon=\gamma_{\rm C}/\chi=0.01\).}
    \label{fig:SM-OAT-dephasing}
\end{figure}

\subsubsection{Diffusive collective dephasing}

Diffusive collective dephasing is described, in the same dimensionless
time \(\tau\), by
\begin{equation}
\frac{d\rho}{d\tau}
=
-i[J_z^2,\rho]
+\epsilon\left(
J_z\rho J_z-\frac12\{J_z^2,\rho\}
\right),
\label{eq:SM-diffusive-master-equation}
\end{equation}
where \(\epsilon\) is the dephasing rate relative to the OAT strength.
The rotation of angle \(\theta\) is encoded after preparation, followed
by an MAI readout with inverted OAT interaction. At \(\theta=0\), the
first moments vanish and the commutator and covariance matrices are~\cite{BaamaraCRPhys2022}
\begin{equation}
\begin{aligned}
C_{\rm dif}
&=
\frac N2
\begin{pmatrix}
\displaystyle
\frac{N-1}{2}
\left(
e^{-5\epsilon\tau/2}
+e^{-\epsilon\tau/2}
\right)
\sin\tau\cos^{N-2}\tau
&
-e^{-\epsilon\tau/2}\cos^{N-1}\tau
\\
e^{-\epsilon\tau} & 0
\end{pmatrix},
\\[1mm]
\Gamma_{\rm dif}
&=
\begin{pmatrix}
\displaystyle
\frac{N(N+1)}8
-\frac{N(N-1)}8e^{-4\epsilon\tau}
&0
\\
0&\displaystyle\frac N4
\end{pmatrix}.
\end{aligned}
\label{eq:SM-diffusive-MAI}
\end{equation}
\end{document}